\documentclass[twocolumn,showpacs,preprintnumbers,amsmath,amssymb]{revtex4}
\usepackage{mathrsfs}
\usepackage{amsfonts}

\usepackage{graphicx}% Include figure files
\usepackage{dcolumn}% Align table columns on decimal point
\usepackage{bm}% bold math
\usepackage{amsmath,amsthm}
\usepackage{overpic}

\usepackage{diagbox}
\begin{document}

\preprint{APS/123-QED}

\title{Characterizing Kirkwood-Dirac nonclassicality and uncertainty diagram based on discrete Fourier transform}% Force line breaks with \\

\author{Ying-Hui Yang$^{1}$}
 \email{yangyinghui4149@163.com}
\author{Bing-Bing Zhang$^{1}$}%
\author{Xiao-Li Wang$^{1}$}%
\author{Shi-Jiao Geng$^{1}$}
\author{Pei-Ying Chen$^{1}$}
\affiliation{%
 $^{1}$School of Mathematics and Information Science, Henan Polytechnic University, Jiaozuo, 454000, China\\
}%

\date{\today}% It is always \today, today,

\begin{abstract}
In this paper, we investigate the Kirkwood-Dirac nonclassicality and uncertainty diagram based on discrete Fourier transform (DFT) in a $d$ dimensional system. The uncertainty diagram of complete incompatibility bases $\mathcal {A},\mathcal {B}$ are characterized by De Bi\`{e}vre [arXiv: 2207.07451]. We show that for the uncertainty diagram of the DFT matrix which is a transition matrix from basis $\mathcal {A}$ to basis $\mathcal {B}$, there is no ``hole" in the region of the $(n_{\mathcal {A}}, n_{\mathcal {B}})$-plane above and on the line $n_{\mathcal {A}}+n_{\mathcal {B}}\geq d+1$,  whether the bases $\mathcal {A},\mathcal {B}$ are not complete incompatible bases or not. Then we present that the KD nonclassicality of a state based on the DFT matrix can be completely characterized by using the support uncertainty relation $n_{\mathcal {A}}(\psi)n_{\mathcal {B}}(\psi)\geq d$, where $n_{\mathcal {A}}(\psi)$ and $n_{\mathcal {B}}(\psi)$ count the number of nonvanishing coefficients in the basis $\mathcal {A}$ and $\mathcal {B}$ representations, respectively. That is, a state $|\psi\rangle$ is KD nonclassical if and only if $n_{\mathcal {A}}(\psi)n_{\mathcal {B}}(\psi)> d$, whenever $d$ is prime or not. That gives a positive answer to the conjecture in [Phys. Rev. Lett. \textbf{127}, 190404 (2021)].
\end{abstract}

\maketitle

%\tableofcontents

\section{Introduction}
In quantum mechanics, there exist many nonclassical properties such as entanglement, discord, coherence, nonlocality, contextuality, noncommutativity of two operators, uncertainty principles, and negativity or nonreality of quasiprobability distributions. By studying these nonclassical properties one can not only obtain a better understanding of quantum mechanics but also explore their applications in quantum information processing. Kirkwood-Dirac (KD) distribution is a quasiprobability distribution that is independently developed by Kirkwood \cite{Kirkwood.1933} and Dirac \cite{Dirac.1945}. It is a finite dimensional analog of the well-known Wigner distribution \cite{Wigner.1932,Wootters.1987}. A quasiprobability distribution behaves like a probability distribution, but negative or nonreal values are allowed to appear in the distribution. For a quantum state and some observables, the KD distribution of this state can be obtained. A quantum state is called KD classical if KD distribution of the state is real nonnegative everywhere, i.e., a probability distribution. Otherwise, it is called KD nonclassical. Recently, KD nonclassicality has come to the forefront due to the application in quantum tomography \cite{Lundeen.2012,Bamber.2014,Thekkadath.2016} and weak measurements \cite{Pusey.2014,Dressel.2015}.

The noncommutativity of observables cannot guarantee the KD nonclassicality of a state. The KD nonclassicality of a state depends not only on the state but also on the eigenbases of observables. Given a state $|\psi\rangle$ and an eigenbasis $\mathcal {A}$ of observable $A$ and an eigenbasis $\mathcal {B}$ of observable $B$, authors in Ref.\cite{Arvidsson-Shukur.2021} gave a sufficient condition on the KD nonclassicality of a state, that is, $|\psi\rangle$ is KD nonclassical if $n_{\mathcal {A}}(\psi)+n_{\mathcal {B}}(\psi)> \lfloor\frac{3d}{2}\rfloor$, where $n_{\mathcal {A}}(\psi)$ counts the number of nonvanishing coefficients in the basis $\mathcal {A}$ representation, similar for $n_{\mathcal {B}}(\psi)$. In 2021, De Bi\`{e}vre \cite{Bievre.2021} introduce the concept of complete incompatibility on eigenbases $\mathcal {A}, \mathcal {B}$ of two observables $A, B$, and presented the relations among complete incompatibility, support uncertainty, and KD nonclassicality, also showed that $|\psi\rangle$ is KD nonclassical if $n_{\mathcal {A}}(\psi)+n_{\mathcal {B}}(\psi)> d+1$ and $\langle a_{i}|b_{j}\rangle\neq 0$, where $|a_{i}\rangle$ and $|b_{j}\rangle$ are the eigenvectors of $\mathcal {A}, \mathcal {B}$, respectively. Xu \cite{Xu.2022} generalized the concept of complete incompatibility to $s$-order incompatibility and established a link between $s$-order incompatibility and the minimal support uncertainty. Recently, De Bi\`{e}vre \cite{Bievre.2022} provided an in-depth study of the links of complete incompatibility to support uncertainty and to KD nonclassicality.

Discrete Fourier transform (DFT) is an important linear transform in quantum information theory. The uncertainty diagram is a practical and visual tool to study the uncertainty of a state with respect to bases $\mathcal {A},\mathcal {B}$. De Bi\`{e}vre \cite{Bievre.2022} characterized the uncertainty diagram of complete incompatibility bases. However, for the DFT matrix with nonprime order, the bases $\mathcal {A},\mathcal {B}$ are not completely incompatible bases. The uncertainty diagram of the DFT matrix with nonprime order is still unclear. In addition, in $d$-dimensional system, all the states are nonclassical except for basis vectors if $d$ is prime \cite{Bievre.2021}. However, for nonprime $d$, it is still an open problem. In this paper, we consider these two questions. Firstly, for the uncertainty diagram of DFT, we show that for any dimension $d$ there is no ``hole" in the region of the $(n_{\mathcal {A}}, n_{\mathcal {B}})$-plane above and on the line $n_{\mathcal {A}}+n_{\mathcal {B}}\geq d+1$, i.e., there is no absence of states with $(n_{\mathcal {A}}(\psi), n_{\mathcal {B}}(\psi))$ in the region. We also show some positions of holes in the uncertainty diagram of DFT. Secondly, we present the KD nonclassicality of a state based on the DFT matrix that is a transition matrix of bases $\mathcal {A},\mathcal {B}$. The KD nonclassicality of a state based on the DFT matrix can be completely characterized by using the support uncertainty relation $n_{\mathcal {A}}(\psi)n_{\mathcal {B}}(\psi)\geq d$. That is, a state $|\psi\rangle$ is KD nonclassical if and only if $n_{\mathcal {A}}(\psi)n_{\mathcal {B}}(\psi)> d$, whenever $d$ is prime or not. In other words, the lower bound, $n_{\mathcal {A}}(\psi)n_{\mathcal {B}}(\psi)= d$, is just attained for the KD classical states. The DFT is an example of a transition matrix for mutually unbiased bases (MUBs). For general MUBs, we also give a sufficient condition on the KD nonclassicality.

The rest of this paper is organized as follows. In Sec. \uppercase\expandafter{\romannumeral2}, we recall some relevant notions and notations. In Sec. \uppercase\expandafter{\romannumeral3} we study the uncertainty diagram of DFT. In Sec. \uppercase\expandafter{\romannumeral4}, we characterize the KD classicality of a state based on DFT and give a sufficient condition of KD nonclassicality for general MUBs. Conclusions are given in Sec. \uppercase\expandafter{\romannumeral5}.

\section{Preliminaries}

\theoremstyle{remark}
\newtheorem{definition}{\indent Definition}
\newtheorem{lemma}{\indent Lemma}
\newtheorem{theorem}{\indent Theorem}
\newtheorem{corollary}{\indent Corollary}

\def\QEDclosed{\mbox{\rule[0pt]{1.3ex}{1.3ex}}}
\def\QED{\QEDclosed}
\def\proof{\indent{\em Proof}.}
\def\endproof{\hspace*{\fill}~\QED\par\endtrivlist\unskip}

Consider a Hilbert space $\mathcal {H}$ with dimension $d$. Let an orthonormal basis $\mathcal {A}=\{|a_{i}\rangle\}_{i=0}^{d-1}$, respectively $\mathcal {B}=\{|b_{j}\rangle\}_{j=0}^{d-1}$, be the eigenbasis of observable $A$, respectively of observable $B$. Let $U$ be the unitary transition matrix with entries $U_{ij}=\langle a_{i}|b_{j}\rangle$ from basis $\mathcal {A}$ to basis $\mathcal {B}$. In terms of these two bases, the Kirkwood-Dirac (KD) distribution of a state $|\psi\rangle \in \mathcal {H}$ can be written as
\begin{equation}
\label{Qij}
Q_{ij}=\langle a_{i}|\psi\rangle\langle\psi|b_{j}\rangle\langle b_{j}|a_{i}\rangle, i,j\in \mathbb{Z}_{d}.
\end{equation}
It is a quasi-probability distribution and satisfies $\sum_{i,j=0}^{d-1}Q_{ij}=1$ with conditional probabilities $Q(a_{i}|\psi)=\sum_{j=0}^{d-1}Q_{ij}=|\langle a_{i}|\psi\rangle|^{2}$ and $Q(b_{j}|\psi)=\sum_{i=0}^{d-1}Q_{ij}=|\langle b_{j}|\psi\rangle|^{2}$. A state $|\psi\rangle$ is called to be classical if the KD distribution of $|\psi\rangle$ is a probability distribution, i.e., $Q_{ij}\geq 0$ for all $i,j \in \mathbb{Z}_{d}$. Otherwise, $|\psi\rangle$ is called to be nonclassical. Obviously, all of the basis vectors $|a_{i}\rangle$ and $|b_{j}\rangle$ are classical.

Given a state $|\psi\rangle \in \mathcal {H}$, let $n_{\mathcal {A}}(\psi)$, respectively $n_{\mathcal {B}}(\psi)$, be the number of nonzero components of $|\psi\rangle$ on $\mathcal {A}$, respectively on $\mathcal {B}$. That is, $n_{\mathcal {A}}(\psi)=|S_{\psi}|$ and $n_{\mathcal {B}}(\psi)=|T_{\psi}|$, where
\begin{equation}
\label{n-ab}
\begin{aligned}
S_{\psi}&=&\{i|\langle a_{i}|\psi\rangle \neq 0, i\in  \mathbb{Z}_{d}\},\\
T_{\psi}&=&\{j|\langle b_{j}|\psi\rangle \neq 0, j\in  \mathbb{Z}_{d}\},
\end{aligned}
\end{equation}
and $|\cdot|$ denotes the cardinality of a set.

Two bases $\mathcal {A}$ and $\mathcal {B}$ are called completely incompatible \cite{Bievre.2021} if all index set $S,T\in \mathbb{Z}_{d}$ for which $|S|+|T|\leq d$ have the property that
\begin{equation}
\mathcal {H}(S,T):=\Pi_{\mathcal {A}}(S)\mathcal {H}\bigcap\Pi_{\mathcal {B}}(T)\mathcal {H}=\{0\},
\end{equation}
where $\Pi_{\mathcal {A}}(S)$ is an orthogonal projector $\sum_{i\in S}|a_{i}\rangle\langle a_{i}|$ and $\Pi_{\mathcal {A}}(S)\mathcal {H}$ is a $|S|$-dimensional subspace. Notice that any $|\psi\rangle\in\Pi_{\mathcal {A}}(S)\mathcal {H}$ implies $S_{\psi}\subseteq S$. If $\mathcal {A}$ and $\mathcal {B}$ are completely incompatible, the only classical states are the basis states \cite{Bievre.2022}.

The uncertainty diagram for orthonormal bases $\mathcal {A},\mathcal {B}$, denoted by UNCD$(\mathcal {A},\mathcal {B})$, is a set of points $(n_{\mathcal {A}},n_{\mathcal {B}})\in \mathbb{Z}_{d+1}^{*}\times \mathbb{Z}_{d+1}^{*}$ in the $n_{\mathcal {A}}n_{\mathcal {B}}$-plane for which there exists a state $|\psi\rangle$ such that $n_{\mathcal {A}}(\psi)=n_{\mathcal {A}}$ and $n_{\mathcal {B}}(\psi)=n_{\mathcal {B}}$, where $\mathbb{Z}_{d+1}^{*}=\mathbb{Z}_{d+1}\setminus \{0\}$. For any state $|\psi\rangle \in \mathcal {H}$, we have $n_{\mathcal {A}}(\psi)n_{\mathcal {B}}(\psi)\geq$ max$_{i,j}|\langle a_{i}|b_{j}\rangle|^{-2}$, where $i,j\in  \mathbb{Z}_{d}$\cite{Bievre.2021,Bievre.2022,Donoho.1989}. It is called the support uncertainty relation. If $\mathcal {A},\mathcal {B}$ are mutually unbiased bases (MUBs)\cite{Planat.2006,Durt.2010}, i.e., $|\langle a_{i}|b_{j}\rangle|=\frac{1}{\sqrt{d}}$ for all $i,j\in  \mathbb{Z}_{d}$, we have $n_{\mathcal {A}}(\psi)n_{\mathcal {B}}(\psi)\geq d$. It means that all the points $(n_{\mathcal {A}},n_{\mathcal {B}})\in$ UNCD$(\mathcal {A},\mathcal {B})$ are above or on the hyperbola $n_{\mathcal {A}}n_{\mathcal {B}}= d$.

The following lemma was introduced in Ref.\cite{Bievre.2022}. It can be employed to determine whether a point $(n_{\mathcal {A}},n_{\mathcal {B}})$ belongs to UNCD$(\mathcal {A},\mathcal {B})$.
\begin{lemma}
\label{exist_Bievre}
Let $S, T$ be two subsets of $\mathbb{Z}_{d}$ and suppose dim $\mathcal {H}(S,T)=L\geq 1$. Suppose that for all $S'\subseteq S$ for which $|S'|=|S|-1$, one has dim $\mathcal {H}(S',T)\leq L-1$, and that for all $T'\subseteq T$ for which $|T'|=|T|-1$, one has dim $\mathcal {H}(S,T')\leq L-1$. Then the set of $|\psi\rangle \in \mathcal {H}(S,T)$ for which $n_{\mathcal {A}}(\psi)=|S|$, $n_{\mathcal {B}}(\psi)=|T|$ is an open and dense set in $\mathcal {H}(S,T)$. The opposite implication is also true.
\end{lemma}

By Lemma \ref{exist_Bievre}, the point $(d, d)$ belongs to the UNCD$(\mathcal {A},\mathcal {B})$ since $\dim\mathcal {H}(S,T)=\dim\mathcal {H}=d$ and $\dim\mathcal {H}(S',T)=\dim\mathcal {H}(S,T')=d-1$.
In order to better employ Lemma \ref{exist_Bievre}, let us first consider subspace $\mathcal {H}(S,T)$. Without loss of generality, let $S=\{0,1,\ldots,k-1\}$, $T=\{0,1,\ldots,l-1\}$ and $|\varphi\rangle=\sum_{j=0}^{l-1}\beta_{j}|b_{j}\rangle$. If $k=d$, then $\mathcal {H}(S,T)=\Pi_{\mathcal {B}}(T)\mathcal {H}$. If $k\neq d$, $\mathcal {H}(S,T)$ is the null space of a $(d-k) \times l$ matrix. Notice that \cite{Bievre.2021}
\begin{equation}
\begin{aligned}
 \mathcal {H}(S,T)
=& \{(\beta_{0},\ldots,\beta_{l-1}) | \langle a_{i}|\varphi\rangle=0, i\in \mathbb{Z}_{d}\setminus S\}. \\
=& \{(\beta_{0},\ldots,\beta_{l-1}) | \sum_{j=0}^{l-1}\langle a_{i}|b_{j}\rangle\beta_{j}=0, i\in \mathbb{Z}_{d}\setminus S\}.
\end{aligned}
\end{equation}
It follows that $\mathcal {H}(S,T)$ is the null space of the matrix $(\langle a_{i}|b_{j}\rangle)_{(d-k) \times l}$. In this paper, a submatrix of a matrix $U$ is denoted by
\begin{eqnarray}
U\left(
\begin{array}{cccc}
i_{0}, i_{1}, \ldots, i_{s}; \\
j_{0}, j_{1}, \ldots, j_{t}.
\end{array}
\right),
\end{eqnarray}
where $i_{k}$ and $j_{l}$ are the $i_{k}$-th row and $j_{l}$-th column of $U$ and $i_{k}, j_{l}\in \mathbb{Z}_{d}$.

Now an improved lemma is given to show the exitance of point $(n_{\mathcal {A}},n_{\mathcal {B}})$ in UNDC$(\mathcal {A}, \mathcal {B})$.

\begin{lemma}
\label{exist_mine}
In UNDC$(\mathcal {A}, \mathcal {B})$, suppose $n_{\mathcal {A}}\neq d$. Then a point $(n_{\mathcal {A}}, n_{\mathcal {B}})\in$ UNDC$(\mathcal {A}, \mathcal {B})$ if and only if there exists a $(d-n_{\mathcal {A}}) \times n_{\mathcal {B}}$ submatrix $M$ of the transition matrix $U$
\begin{eqnarray}
M=U\left(
\begin{array}{cccc}
i_{0}, i_{1}, \ldots, i_{d-n_{\mathcal {A}}-1}; \\
j_{0}, j_{1}, \ldots, j_{n_{\mathcal {B}}}.
\end{array}
\right),
\end{eqnarray}
which satisfies the following three conditions:

(i) Rank$(M) < n_{\mathcal {B}}$;

(ii) Rank$(M')=$Rank$(M)+1$, where
\begin{eqnarray}
\label{M'}
M'=U\left(
\begin{array}{cccc}
i_{0}, i_{1}, \ldots, i_{d-n_{\mathcal {A}}-1}, i_{d-n_{\mathcal {A}}}; \\
j_{0}, j_{1}, \ldots, j_{n_{\mathcal {B}}}.
\end{array}
\right)
\end{eqnarray}
and $i_{d-n_{\mathcal {A}}} \in \mathbb{Z}_{d}\setminus\{i_{0}, i_{1}, \ldots, i_{d-n_{\mathcal {A}}-1}\}$;

(iii) Rank$(M'')=$Rank$(M)$, where $M''$ is a new submatrix of $U$ that is obtained by removing one column of $M$.
\end{lemma}

The proof of Lemma \ref{exist_mine} is given in Appendix \ref{Prove exist_mine}. Note that $M'$ in Eq.(\ref{M'}) is a submatrix of $U$ and condition (ii) means the rank will increase by one if a new row is added to the submatrix $M$. And condition (iii) means the rank is invariant if a column of $M$ is removed.

Now we introduce the discrete Fourier transform (DFT). Suppose that $F$ is the DFT matrix with $F=(F_{ij})=(\frac{1}{\sqrt{d}}\omega_{d}^{ij})$, where $i,j\in \mathbb{Z}_{d}$ and $\omega_{d}=e^{2\pi \sqrt{-1}/d}$. Obviously, $F$ is a symmetric and reversible Vandermonde matrix. The DFT matrix $F$ has the following property.

\begin{lemma}
\label{properties on DFT}
Suppose $m | d$ but $m\neq d$. Let
\begin{eqnarray}
\label{M-(i)}
M=&F\left(
\begin{array}{cccc}
 i_{0}, i_{0}+m, \cdots, i_{0}+(t-1)m;\\
 j_{0}, j_{1}, \cdots, j_{s-1}.
\end{array}
\right),
\end{eqnarray}
where $t\leq \frac{d}{m}$ and $j_{l}\neq j_{k} \mod \frac{d}{m}$ for $l\neq k$. Then Rank($M$)$=\min\{s,t\}$.
%(ii) For even $d$, the following two submatrices $M, N$ are equal, i.e., $M=N$, and Rank$(M)=s$, where
%\begin{eqnarray}
%\label{even-M-(ii)}
%M=&\left(
%\begin{array}{cccc}
% i_{0}, i_{1}, \cdots, i_{s-1};\\
% 0, 2, \cdots, 2(t-1).
%\end{array}
%\right),\\
%\label{even-N-(ii)}
%N=&\left(
%\begin{array}{cccc}
% \frac{d}{2}+i_{0}, \frac{d}{2}+i_{1}, \cdots, \frac{d}{2}+i_{s-1};\\
% 0, 2, \cdots, 2(t-1).
%\end{array}
%\right),
%\end{eqnarray}
%$i_{k}\in \mathbb{Z}_{\frac{d}{2}}$, $i_{l}\neq i_{k}$ for $l\neq k$ and $s\leq t\leq \frac{d}{2}$.

%(ii) For the following submatrix
%\begin{eqnarray}
%\label{M-(iii)}
%M=&\left(
%\begin{array}{cccc}
% 0, 1, \cdots, s-1;\\
% 0, m, \cdots, (t-1)m.
%\end{array}
%\right)
%\end{eqnarray}
%with $m|d$ and $s<t\leq \frac{d}{m}$,if row $i_{l}$ of $F$ is added to $M$ to obtain submatrix $M'$, then Rank($M'$)=Rank($M$)+1, where $i_{l}\in \{s, \cdots, \frac{d}{m}-1\}$;
%If any column of $M$ is removed to obtain submatrix $M''$, then Rank($M''$)=Rank($M$).
\end{lemma}
The proof of Lemma \ref{properties on DFT} is given in Appendix \ref{Prove properties on DFT}.
Since $F$ is symmetric, a similar property can be obtained if one interchanges indices of the rows with that of columns of $M$ in Eq.(\ref{M-(i)}). Lemma \ref{properties on DFT} means that $M$ in Eq.(\ref{M-(i)}) is a row full rank matrix or a column full rank matrix.

\section{Uncertainty diagram of DFT}
%Although Theorem \ref{nonclassicality on DFT} completely characterizes the KD nonclassicality of a state based on DFT, it does not imply that in a uncertainty diagram UNDC$(\mathcal {A}, \mathcal {B})$, all the points $(n_{\mathcal {A}}, n_{\mathcal {B}})$ above or on the hyperbola of $n_{\mathcal {A}}(\psi)n_{\mathcal {B}}(\psi) = d$ do exist.

De Bi\`{e}vre \cite{Bievre.2021} has shown that the points on the hyperbola $n_{\mathcal {A}}(\psi)n_{\mathcal {B}}(\psi) = d$ belong to UNDC$(\mathcal {A}, \mathcal {B})$ of the DFT matrix $F$. He \cite{Bievre.2022} also showed that $\mathcal {A}$ and $\mathcal {B}$ are completely incompatible if and only if UNDC$(\mathcal {A}, \mathcal {B})=\{(n_{\mathcal {A}},n_{\mathcal {B}})|n_{\mathcal {A}}+n_{\mathcal {B}}\geq d+1,\mathcal {A}, \mathcal {B}\in \mathbb{Z}_{d+1}^{*}\}$. However, it is unclear if $\mathcal {A}$ and $\mathcal {B}$ are not completely incompatible. In this section, we will continue to explore UNDC$(\mathcal {A}, \mathcal {B})$ of $F$.

The UNDC$(\mathcal {A}, \mathcal {B})$ of $F$ is symmetric since $F$ is symmetric \cite{Bievre.2022}. That is, $(n_{\mathcal {A}},n_{\mathcal {B}})\in$ UNDC$(\mathcal {A}, \mathcal {B})$ of $F$ if and only if $(n_{\mathcal {B}},n_{\mathcal {A}})\in$ UNDC$(\mathcal {A}, \mathcal {B})$ of $F$.

\begin{theorem}
\label{UNCD-main}
Suppose $m|d$ and $n\neq 0$. A point $(d-n,n_{\mathcal {B}})$ belongs to UNDC$(\mathcal {A}, \mathcal {B})$ of $F$ if $m|n$ and $\frac{n}{m} < n_{\mathcal {B}} \leq \frac{d}{m}$.
\end{theorem}
\begin{proof}
In order to show $(d-n,n_{\mathcal {B}})\in$ UNDC$(\mathcal {A}, \mathcal {B})$, we only need to find a $n \times n_{\mathcal {B}}$ submatrix that satisfies Lemma \ref{exist_mine}. First of all, let
\begin{eqnarray}
\label{}
N=&F\left(
\begin{array}{cccc}
 0, 1, \cdots, \frac{n}{m}-1;\\
 0, m, \cdots, (n_{\mathcal {B}}-1)m.
\end{array}
\right)
=\left(
\begin{array}{c}
\omega_{d}^{ikm}
\end{array}
\right)_{\frac{n}{m}\times n_{\mathcal {B}}},\nonumber
\end{eqnarray}
where $i \in \mathbb{Z}_{\frac{n}{m}}$ and $k \in \mathbb{Z}_{n_{\mathcal {B}}}$. Obviously, $N$ be a $\frac{n}{m}\times n_{\mathcal {B}}$ submatrix of $F$. Since $\frac{n}{m} < n_{\mathcal {B}}$, $N$ is a Vandermonde Matrix that is a row full rank matrix by Lemma \ref{properties on DFT}.
The matrix $N$ has the following two properties.

(i) If row $i_{1} \in \{\frac{n}{m},\frac{n}{m}+1,\ldots,\frac{d}{m}-1\}$ of $F$ is added to $N$ to obtain submatrix $N'$, then Rank($N'$)=Rank($N$)+1=$\frac{n}{m}+1$ by Lemma \ref{properties on DFT}. It is because $N'$ is still a Vandermonde matrix and  $\frac{n}{m} +1 \leq n_{\mathcal {B}}$ and $\omega_{d}^{im}\neq \omega_{d}^{i'm}$, where $i',i\in \mathbb{Z}_{\frac{n}{m}}\bigcup \{i_{1}\}$, $i'\neq i$.

(ii) If a column of $N$ is removed to obtain submatrix $N''$, then Rank($N''$)=Rank($N$)$=\frac{n}{m}$ by Lemma \ref{properties on DFT} and $\frac{n}{m} \leq n_{\mathcal {B}}-1$.

Secondly, consider the following $n\times n_{\mathcal {B}}$ submatirx
\begin{eqnarray}
M&=&F\left(
\begin{array}{cccc}
 \cdots,\frac{d}{m}i,\frac{d}{m}i+1,\cdots, \frac{d}{m}i+\frac{n}{m}-1,\cdots;\\
 0, m, \cdots, (n_{\mathcal {B}}-1)m.
\end{array}
\right),\nonumber\\
&=&\left(
\begin{array}{c}
N \\
N \\
\vdots\\
N
\end{array}
\right), i= 0,1,\ldots,m-1.
\label{M-d/m}
\end{eqnarray}
The equality in Eq.(\ref{M-d/m}) holds due to $\omega_{d}^{(\frac{d}{m}i+j)\times km}=\omega_{d}^{j\times km}$. It follows
Rank($M$)=Rank($N$)=$\frac{n}{m} < n_{\mathcal {B}}$.

Since $N$ has above two properties and $M$ has the form in Eq.(\ref{M-d/m}),
we have Rank$(M')=\frac{n}{m}+1=$Rank$(M)+1$ and Rank$(M'')=\frac{n}{m}=$Rank$(M)$, where $M'$ is obtained by adding a new row $\frac{d}{m}j+i_{1}$ to $M$, $i_{1} \in \mathbb{Z}_{\frac{d}{m}}\backslash \mathbb{Z}_{\frac{n}{m}}$ and $M''$ is obtained by removing a column of $M$. Here we employ the condition $\omega_{d}^{(\frac{d}{m}j+i_{1})km}=\omega_{d}^{i_{1}km}$.
So $M$ is the required submatrix. By Lemma \ref{exist_mine}, the required result is obtained.
\end{proof}

Notice that $n\neq 0$. It means Theorem \ref{UNCD-main} cannot work for $(d,n_{\mathcal {B}})$. However, taking $m=1$ and $n_{\mathcal {B}}=d$ in Theorem \ref{UNCD-main}, we obtain $(i, d)\in$UNDC$(\mathcal {A}, \mathcal {B})$, where $i\in \mathbb{Z}_{d}^{*}$. By the symmetry of UNDC$(\mathcal {A}, \mathcal {B})$ of $F$, we have $(d, i)\in$UNDC$(\mathcal {A}, \mathcal {B})$, where $i\in \mathbb{Z}_{d}^{*}$. In addition, $(d,d)\in$UNDC$(\mathcal {A}, \mathcal {B})$ by Lemma \ref{exist_Bievre}.

Taking $m=1$ and $n_{\mathcal {A}}=d-n$, we have the following result by Theorem \ref{UNCD-main} and discussions above.
\begin{corollary}
\label{existance on d+1}
A point $(n_{\mathcal {A}}, n_{\mathcal {B}})\in$ UNDC$(\mathcal {A}, \mathcal {B})$ of $F$ if $n_{\mathcal {A}}+n_{\mathcal {B}}\geq d+1$.
\end{corollary}

Note that in Corollary \ref{existance on d+1}, $n_{\mathcal {A}}$ can run over set $\mathbb{Z}_{d+1}^{*}$ due to the above discussion of Corollary \ref{existance on d+1}. Corollary \ref{existance on d+1} means that all the points above and on the line segment $n_{\mathcal {A}}+n_{\mathcal {B}}= d+1$ do exist whether $\mathcal {A}$ and $\mathcal {B}$ are completely incompatible or not. It implies that there is no ``hole" in the region of the $(n_{\mathcal {A}}, n_{\mathcal {B}})$-plane above and on the line $n_{\mathcal {A}}+n_{\mathcal {B}}\geq d+1$ for any $d$, that is, there is no absence of states with $(n_{\mathcal {A}}(\psi), n_{\mathcal {B}}(\psi))$ in the region. The absence of states lies strictly above the hyperbola of $n_{\mathcal {A}}(\psi)n_{\mathcal {B}}(\psi) = d$ and strictly below the line $n_{\mathcal {A}}+n_{\mathcal {B}}= d+1$. This is illustrated in Fig.\ref{Fig1}. The following theorems will show where the holes are.

\begin{theorem}
\label{nonexistance-nb=2}
A point $(d-n,2)$ belongs to UNDC$(\mathcal {A}, \mathcal {B})$ of $F$ if and only if $n=0$ or $n | d$ but $n\neq d$.
\end{theorem}
\begin{proof}
Sufficiency can be obtained by taking $n=m$ in Theorem \ref{UNCD-main} and by discussion above Corollary \ref{existance on d+1} for $n=0$. We now show the necessary. Since $(d,2)$ and $(d-1, 2)$ belong to UNDC$(\mathcal {A}, \mathcal {B})$, we have $n=0,1$, respectively. Then we consider $n\geq 2$. A point $(d-n, 2)$ belongs to UNDC$(\mathcal {A}, \mathcal {B})$ of $F$. It means there exists a $n \times 2$ submatrix
\begin{eqnarray}
\label{}
M=&\left(
\begin{array}{cccc}
 \omega_{d}^{i_{0}j} & \omega_{d}^{i_{0}k}\\
 \cdots & \cdots\\
 \omega_{d}^{i_{n-1}j} & \omega_{d}^{i_{n-1}k}
\end{array}
\right).\nonumber
\end{eqnarray}
satisfying Lemma \ref{exist_mine}. Then we have Rank$(M)$=1. It means that $\omega_{d}^{(i_{t}-i_{s})j}=\omega_{d}^{(i_{t}-i_{s})k}$ for any $s,t\in \mathbb{Z}_{n}$ and $s\neq t$. It follows that $(i_{t}-i_{s})(k-j) \equiv 0 \mod d$. Assume $\gcd(k-j, d)=p$. Then $i_{s}\equiv i_{t} \mod \frac{d}{p}$. That is, $i_{s}$ is in the congruence class of $i_{t}$ modulo $\frac{d}{p}$. Notice that the cardinality of the congruence class of $i_{t}$ modulo $\frac{d}{p}$ is $p$. It implies that there are at most $p$ rows in submatrix $M$ by the arbitrariness of $s,t\in \mathbb{Z}_{n}$, i.e., $n \leq p$. In fact, the submatrix $M$ has only $p$ rows, i.e., $n = p$. Otherwise, $M$ cannot satisfy the second condition of Lemma \ref{exist_mine}. Thus, $n | d$ since $p | d$. However, $n=p\neq d$ since $j\neq k \mod d$.
\end{proof}

Note that for UNDC$(\mathcal {A}, \mathcal {B})$ of $F$, $n_{\mathcal {A}}(\psi)n_{\mathcal {B}}(\psi)\geq d$ for any state $|\psi\rangle \in \mathcal {H}$. Thus $n=d$ is meaningless in Theorem \ref{nonexistance-nb=2}. It means that a point $(d-n,2)\notin$ UNDC$(\mathcal {A}, \mathcal {B})$ of $F$ if and only if $n \nmid d$. In fact, we only consider the case $n\leq \frac{d}{2}$ since $(d-n)\times 2 \geq d$. For example, if $d=6$, there is no hole for $n_{\mathcal {B}}=2$ since $n=1,2,3$ are all the divisors of six. See panel (a) of Fig.\ref{Fig1}. If $d=8$, there is a hole $(5,2)$ since $n=3 \nmid 8$. See panel (b) of Fig.\ref{Fig1}.

\begin{theorem}
\label{nonexistance-nb=3}
Suppose $d$ has only nontrivial prime divisors. Then a point $(d-n,3)$ belongs to UNDC$(\mathcal {A}, \mathcal {B})$ of $F$ if and only if $n=0$ or there exists a divisor $m$ of $d$ such that $3m \leq d$ and  $n=m$ or $n=2m$.
\end{theorem}

The proof of Theorem \ref{nonexistance-nb=3} is given in Appendix \ref{Prove nonexistance-nb=3}.
Theorem \ref{nonexistance-nb=3} presents where the holes are when $n_{\mathcal {B}}=3$. For instance, if $d=9$, $m=1,2$. Then $n=0,1,2,3,6$. It implies that $n_{\mathcal {A}}=9,8,7,6,3$ but $n_{\mathcal {A}}\neq 5, 4$. Thus there is two holes $(5,3)$ and $(4,3)$ for $n_{\mathcal {B}}=3$. See panel (c) of Fig.\ref{Fig1}. Panel (d) is for $d=10$. It should be noted that although $8$ has a nontrivial nonprime divisor, one can check that the result of Theorem \ref{nonexistance-nb=3} for $d=8$ still holds since the proof can be followed similar for the proof of Theorem \ref{nonexistance-nb=3}.

%\begin{widetext}
\begin{figure}[h]
\centering
\begin{overpic}[scale=0.28]{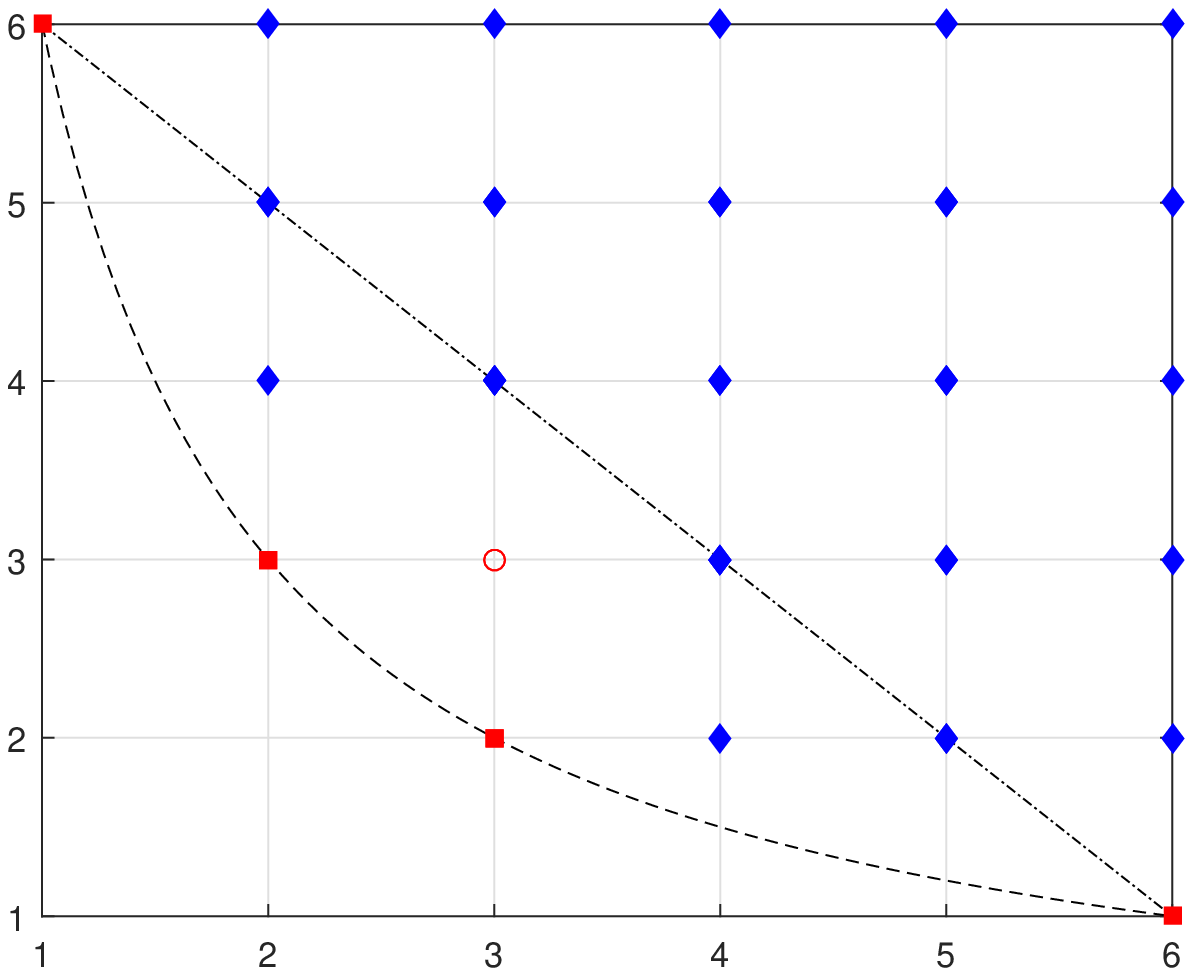}
    \put(12,12){\scriptsize {\bf (a): $d=6$}}
    \put(40,0){\scriptsize { $n_{\mathcal {A}}$}}
    \put(0,40){\scriptsize { $n_{\mathcal {B}}$}}
\end{overpic}
\begin{overpic}[scale=0.28]{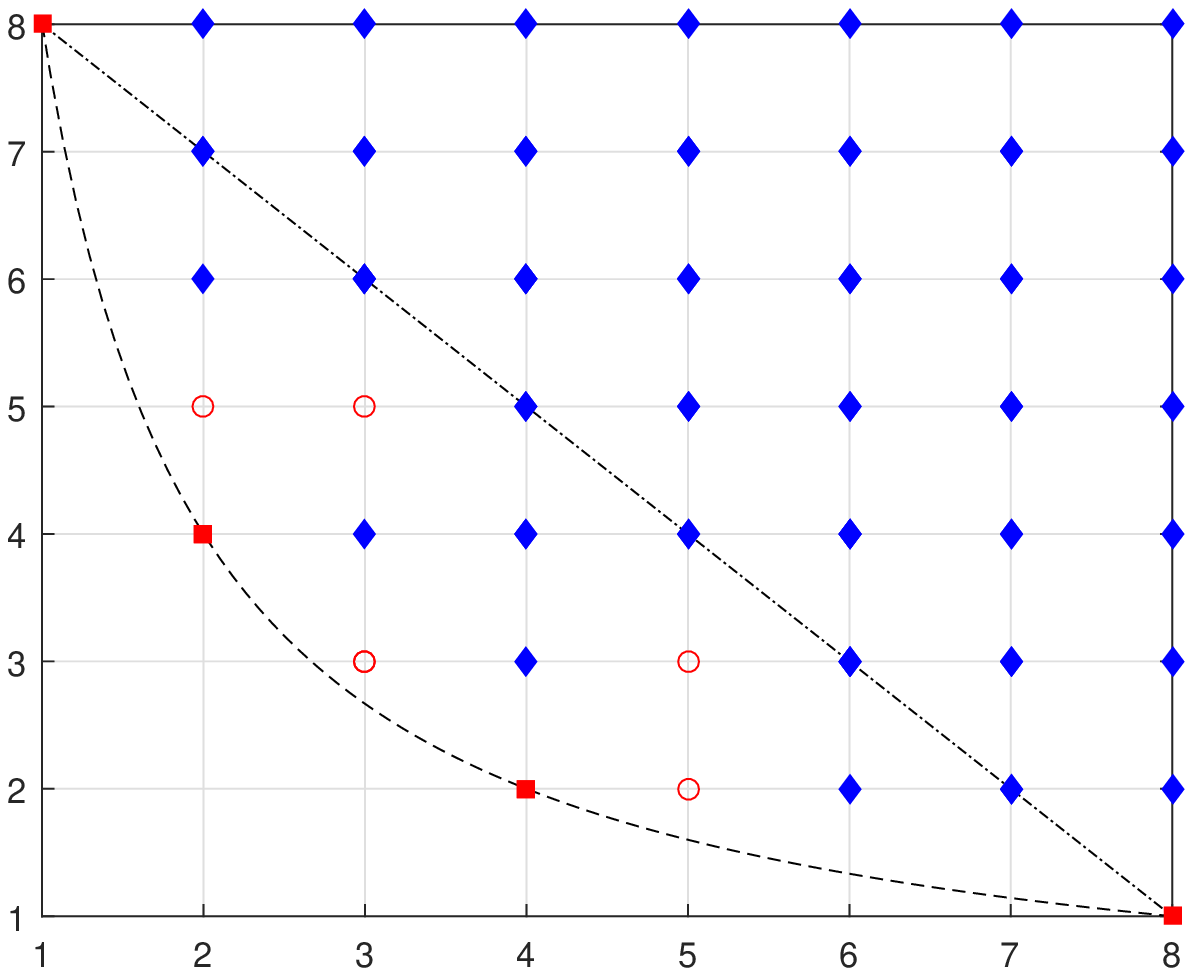}
    \put(12,12){\scriptsize {\bf (b): $d=8$}}
    \put(40,0){\scriptsize { $n_{\mathcal {A}}$}}
    \put(0,40){\scriptsize { $n_{\mathcal {B}}$}}
\end{overpic}
\begin{overpic}[scale=0.28]{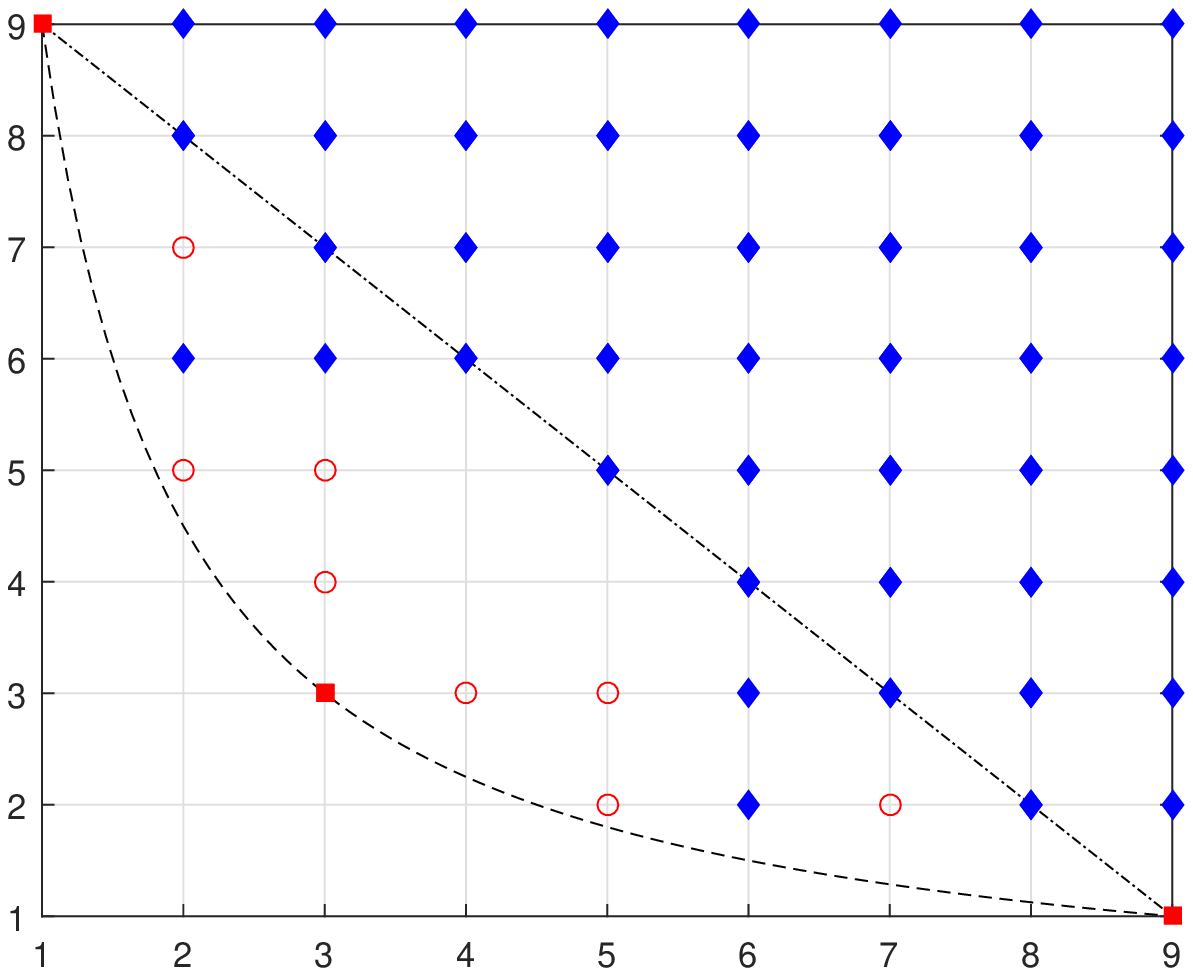}
    \put(12,12){\scriptsize {\bf (c): $d=9$}}
    \put(40,0){\scriptsize { $n_{\mathcal {A}}$}}
    \put(0,40){\scriptsize { $n_{\mathcal {B}}$}}
\end{overpic}
\begin{overpic}[scale=0.28]{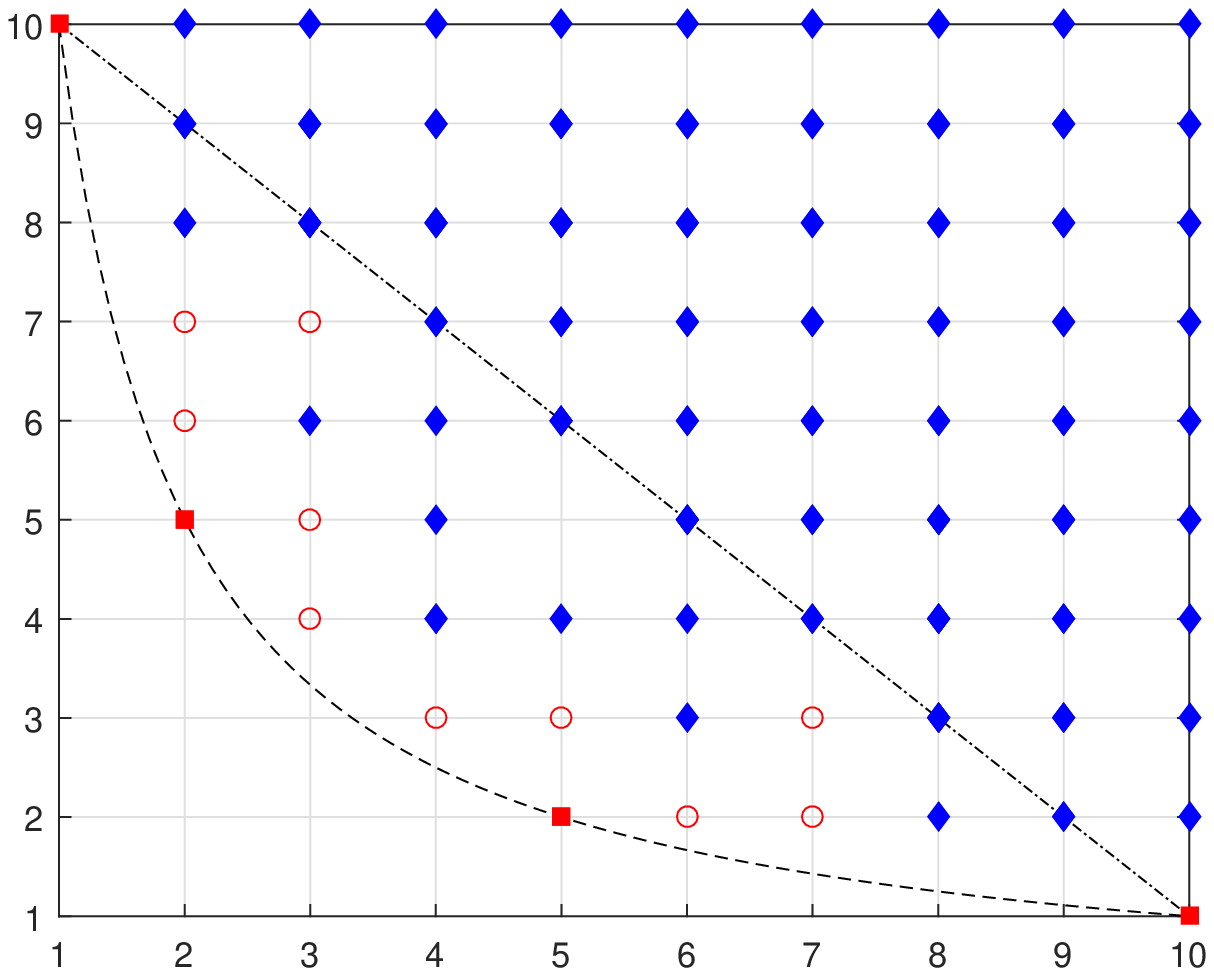}
    \put(12,12){\scriptsize {\bf (d): $d=10$}}
    \put(40,0){\scriptsize { $n_{\mathcal {A}}$}}
    \put(0,40){\scriptsize { $n_{\mathcal {B}}$}}
\end{overpic}
\caption{Uncertainty diagrams of the DFT matrix. (a): $d=6$; (b): $d=8$; (c): $d=9$; (d): $d=10$. Dashed curve is the hyperbola $n_{\mathcal {A}}(\psi)n_{\mathcal {B}}(\psi) = d$. Dot-dashed line is the line $n_{\mathcal {A}}+n_{\mathcal {B}}= d+1$. Red squares represent KD classical states. Blue diamonds represent KD nonclassical states. Red circles represent holes. Blanks mean the existence of the point in UNDC$(\mathcal {A}, \mathcal {B})$ is unclear.}
\label{Fig1}
\end{figure}
%\end{widetext}

\section{KD nonclassicality on MUB}
In this section, we focus on the KD nonclassicality of a state based on MUBs. In Ref.\cite{Bievre.2021}, De Bi\`{e}vre gave a conjecture, that is, whether it is true that the only KD classical states for the DFT are the ones on the hyperbola $n_{\mathcal {A}}(\psi)n_{\mathcal {B}}(\psi) = d$ ? So we first consider the transition matrix of a pair of MUBs $\mathcal {A}$ and $\mathcal {B}$ is the DFT matrix $F$ and try to answer this question.

%Let the DFT matrix $F$ be the unitary transition operator between $\mathcal {A}$ and $\mathcal {B}$ in a $d$ dimensional Hilbert space $\mathcal {H}$.
\begin{theorem}
\label{nonclassicality on DFT}
A state $|\psi\rangle \in \mathcal {H}$ is KD nonclassical if and only if $n_{\mathcal {A}}(\psi)n_{\mathcal {B}}(\psi) > d$. In other word, $|\psi\rangle$ is KD classical if and only if $n_{\mathcal {A}}(\psi)n_{\mathcal {B}}(\psi) = d$.
\end{theorem}
\begin{proof}
The necessity has been proved by De Bi\`{e}vre in Ref.\cite{Bievre.2021}. Here we only need to show the sufficiency, i.e., $|\psi\rangle \in \mathcal {H}$ is KD nonclassical if $n_{\mathcal {A}}(\psi)n_{\mathcal {B}}(\psi) > d$.

We proceed by contradiction. Suppose that $|\psi\rangle$ is KD classical, i.e., $\langle a_{i}|\psi\rangle\langle \psi|b_{j}\rangle \langle b_{j}|a_{i}\rangle \geq 0$ for any $i,j\in \mathbb{Z}_{d}$. Since KD distribution is insensitive to global phase rotations, we perform global phase rotations $|a_{i}\rangle \rightarrow e^{\sqrt{-1}\phi_{i}}|a_{i}\rangle$ and $|b_{j}\rangle| \rightarrow e^{\sqrt{-1}\phi_{j}}|b_{j}\rangle|$ such that $\langle a_{i}|\psi\rangle$ and $\langle \psi|b_{j}\rangle$ are nonnegative for $i,j\in \mathbb{Z}_{d}$. Reordering the basis vectors, we can suppose that $\langle a_{i_{m}}|\psi\rangle >0$ and $\langle \psi|b_{j_{s}}\rangle >0$ for $m\in \mathbb{Z}_{n_{\mathcal {A}}(\psi)}$ and $s\in \mathbb{Z}_{n_{\mathcal {B}}(\psi)}$, where $i_{m},j_{s}$ are initial indices of basis vectors $|a_{i}\rangle$ and $|b_{j}\rangle$, respectively. Thus, for the same range of $i_{m}$ and $j_{s}$, we have $\langle b_{j_{s}}|a_{i_{m}}\rangle \geq 0$. Since $F$ is the DFT matrix, we have $|F_{ij}|=|\langle a_{i}|b_{j}\rangle|=\frac{1}{\sqrt{d}}|\omega_{d}^{ij}|=\frac{1}{\sqrt{d}}$. It follows $\langle a_{i_{m}}| b_{j_{s}}\rangle= \frac{1}{\sqrt{d}}$ for $m\in \mathbb{Z}_{n_{\mathcal {A}}(\psi)}$ and $s\in \mathbb{Z}_{n_{\mathcal {B}}(\psi)}$. It means that the top left-hand block $V=(v_{ij})$ in the new transition matrix after reordering the basis vectors is a $n_{\mathcal {A}} \times n_{\mathcal {B}}$ submatrix with all entries of $\frac{1}{\sqrt{d}}$.

Let us first consider the trivial case. If $n_{\mathcal {A}}(\psi)=1$ (or $n_{\mathcal {B}}(\psi)=1$), then $n_{\mathcal {B}}(\psi)=d$ (or $n_{\mathcal {A}}(\psi)=d$) since $F$ is the DFT matrix. Thus $n_{\mathcal {A}}(\psi)n_{\mathcal {B}}(\psi) = d$. It is a contradiction with $n_{\mathcal {A}}(\psi)n_{\mathcal {B}}(\psi) > d$.

Next we will consider the case $n_{\mathcal {A}}(\psi)\geq2$ and $n_{\mathcal {B}}(\psi)\geq2$.
For any $m\in \mathbb{Z}_{n_{\mathcal {A}}(\psi)}$, any $s,t\in \mathbb{Z}_{n_{\mathcal {B}}(\psi)}$ and $s<t$, calculate the product of two numbers $\sqrt{d}v^{*}_{ms}$ and $\sqrt{d}v_{mt}$ in $V$. We have
\begin{equation}
\label{product of vv}
\begin{aligned}
1=&dv^{*}_{ms} \cdot v_{mt}\\
=&de^{-\sqrt{-1}(\phi_{j_{s}}-\phi_{i_{m}})}\langle b_{j_{s}}|a_{i_{m}}\rangle e^{\sqrt{-1}(\phi_{j_{t}}-\phi_{i_{m}})}\langle a_{i_{m}}|b_{j_{t}}\rangle\\
=&e^{\sqrt{-1}(\phi_{j_{t}}-\phi_{j_{s}})}\omega_{d}^{-i_{m}j_{s}}\omega_{d}^{i_{m}j_{t}}\\
=&\omega_{d}^{\alpha_{s,t}+i_{m}(j_{t}-j_{s})}.
\end{aligned}
\end{equation}
where $\alpha_{s,t}:=\frac{d}{2\pi}(\phi_{j_{t}}-\phi_{j_{s}})$. It implies that $\alpha_{s,t}+i_{m}(j_{t}-j_{s})\equiv 0 \mod d$. Notice that $\alpha_{s,t}$ is independent of $m$. Thus, for any $m, n\in \mathbb{Z}_{n_{\mathcal {A}}(\psi)}$ and $m<n$, we have
\begin{equation}
\label{}
\begin{aligned}
\alpha_{s,t}+i_{m}(j_{t}-j_{s}) &\equiv& 0 \mod d,  \\
\alpha_{s,t}+i_{n}(j_{t}-j_{s}) &\equiv& 0 \mod d.
\end{aligned}
\end{equation}
It follows that
\begin{equation}
\label{}
(i_{n}-i_{m})(j_{t}-j_{s}) \equiv 0 \mod d.
\end{equation}
Suppose gcd$(j_{t}-j_{s},d)=p$ and $q:=\frac{d}{p}$. If $p=1$, then $i_{n} \equiv i_{m} \mod d$. It is impossible since $m, n\in \mathbb{Z}_{n_{\mathcal {A}}(\psi)}$ and $m<n$.
If $p\neq 1$, then we have
\begin{equation}
\label{}
\begin{aligned}
i_{n} &\equiv& i_{m}  \mod q,  \\
j_{t} &\equiv& j_{s}  \mod p.
\end{aligned}
\end{equation}
It implies that $i_{n}$ is in the congruence class of $i_{m}$ modulo $q$ and the cardinality of the congruence class of $i_{m}$ is $p$. Similarly, $j_{t}$ is in the congruence class of $j_{s}$ modulo $p$, and the cardinality of the congruence class of $j_{s}$ is $q$. Because of the arbitrariness of $i_{n}, i_{m} \in \mathbb{Z}_{n_{\mathcal {A}}(\psi)}$ and $j_{t}, j_{s} \in \mathbb{Z}_{n_{\mathcal {B}}(\psi)}$, we obtain $n_{\mathcal {A}} \leq p$ and $n_{\mathcal {B}} \leq q$. Therefore, $n_{\mathcal {A}}(\psi)n_{\mathcal {B}}(\psi) \leq d$. It is a contradiction with $n_{\mathcal {A}}(\psi)n_{\mathcal {B}}(\psi) > d$.
\end{proof}

From the above proof we find that $n_{\mathcal {A}}(\psi)n_{\mathcal {B}}(\psi) \leq d$ if $|\psi\rangle$ is KD classical.  It follows that $n_{\mathcal {A}}(\psi)n_{\mathcal {B}}(\psi) = d$ since $n_{\mathcal {A}}(\psi)n_{\mathcal {B}}(\psi) \geq d$ for the DFT matrix. It implies that only the KD classical states lie on the hyperbola of $n_{\mathcal {A}}(\psi)n_{\mathcal {B}}(\psi) = d$. This result gives a positive answer to the conjecture in Ref.\cite{Bievre.2021}.

When $d$ is prime, the bases $\mathcal {A},\mathcal {B}$ are completely incompatible. Then all the states are nonclassical except for basis vectors \cite{Bievre.2021}. In Theorem \ref{nonclassicality on DFT}, the KD nonclassicality of a state based on the DFT matrix, whenever $d$ is prime or not, is completely characterized by the support uncertainty relation. See Fig.\ref{Fig1}.
%Although there are some points in UNDC$(\mathcal {A}, \mathcal {B})$.

For two mutually unbiased observables, the author in Ref.\cite{Bievre.2022} showed that all states, except for the eigenstates of the two observables, are KD nonclassical when two observables are completely incompatible. Next, we will consider the general case, that is, two observables are not necessarily completely incompatible.

\begin{theorem}
\label{nonclassicality on one}
Let $U$ be the unitary transition operator between $\mathcal {A}$ and $\mathcal {B}$ in a $d$ dimensional Hilbert space $\mathcal {H}$ and suppose that $|U_{ij}|=|\langle a_{i}|b_{j}\rangle|=\frac{1}{\sqrt{d}}$ for $i,j\in \mathbb{Z}_{d}$. $|\psi\rangle \in \mathcal {H}$ but not a basis vector is KD nonclassical if $n_{\mathcal {A}}(\psi) > \frac{d}{2}$ or $n_{\mathcal {B}}(\psi) > \frac{d}{2}$.
\end{theorem}
\begin{proof}
Suppose that $|\psi\rangle$ is KD classical. Similar to the proof of Theorem \ref{nonclassicality on DFT}, we perform global phase rotations and reorder and relabel the basis vectors such that $\langle a_{i}|\psi\rangle >0$ and $\langle \psi|b_{j}\rangle >0$ for $i\in \mathbb{Z}_{n_{\mathcal {A}}(\psi)}$ and $j\in \mathbb{Z}_{n_{\mathcal {B}}(\psi)}$. Thus, $\langle b_{j}|a_{i}\rangle \geq 0$ for the same range of $i$ and $j$. Since $|U_{ij}|=|\langle a_{i}|b_{j}\rangle|=\frac{1}{\sqrt{d}}$, we have $\langle a_{i}| b_{j}\rangle= \frac{1}{\sqrt{d}}$ for $i\in \mathbb{Z}_{n_{\mathcal {A}}(\psi)}$ and $j\in \mathbb{Z}_{n_{\mathcal {B}}(\psi)}$.

Note that it is impossible that $n_{\mathcal {A}}(\psi)=1$ or $n_{\mathcal {B}}(\psi)=1$ since $|\psi\rangle$ is not a basis vector.  For any $0\leq j<j'\leq \mathbb{Z}_{n_{\mathcal {B}}(\psi)}$, we have
\begin{equation}
\label{}
\begin{aligned}
0=&d\langle b_{j}|b_{j'}\rangle \\
 =&d\sum_{i=0}^{n_{\mathcal {A}}-1}\langle b_{j}|a_{i}\rangle \langle a_{i}|b_{j'}\rangle+d\sum_{i=n_{\mathcal {A}}}^{d-1}\langle b_{j}|a_{i}\rangle \langle a_{i}|b_{j'}\rangle\\
 =& n_{\mathcal {A}}(\psi)+P.
\end{aligned}
\end{equation}
where $P:=d\sum_{i=n_{\mathcal {A}}}^{d-1}\langle b_{j}|a_{i}\rangle \langle a_{i}|b_{j'}\rangle$. It follows that $P=-n_{\mathcal {A}}(\psi)$.

If $n_{\mathcal {A}}(\psi) > \frac{d}{2}$, we have $|P|>\frac{d}{2}$. While
\begin{equation}
\label{}
\begin{aligned}
|P|\leq & d\sqrt{\sum_{i=n_{\mathcal {A}}}^{d-1}|\langle b_{j}|a_{i}\rangle|^{2}\sum_{i=n_{\mathcal {A}}}^{d-1}|\langle a_{i}|b_{j'}\rangle|^{2}} \\
 =&d-n_{\mathcal {A}}\\
 <& \frac{d}{2}.
\end{aligned}
\end{equation}
It is a contradiction. Note that the first inequality follows from Cauchy-Schwarz inequality. This contradiction implies that $|\psi\rangle \in \mathcal {H}$ is KD nonclassical.

The case when $n_{\mathcal {B}}(\psi) > \frac{d}{2}$ can be proved similarly by analyzing $\langle a_{i}|a_{i'}\rangle$ for any $0\leq i<i'\leq \mathbb{Z}_{n_{\mathcal {A}}(\psi)}$.
\end{proof}

Theorem \ref{nonclassicality on one} gives a sufficient condition of KD nonclassicality for two mutually unbiased observables, whether they are completely incompatible or not. However, it is still unclear that $|\psi\rangle \in \mathcal {H}$ is KD nonclassical or not if $n_{\mathcal {A}}(\psi) \leq \frac{d}{2}$ and $n_{\mathcal {B}}(\psi) \leq \frac{d}{2}$.

\section{Conclusion}
We have studied the uncertainty diagram and the Kirkwood-Dirac nonclassicality based on DFT in a $d$ dimensional system. We show that for the uncertainty diagram of the DFT matrix, there is no ``hole" in the region of the $(n_{\mathcal {A}}, n_{\mathcal {B}})$-plane above and on the line $n_{\mathcal {A}}+n_{\mathcal {B}}\geq d+1$, whether the bases $\mathcal {A},\mathcal {B}$ are not complete incompatible bases or not. The absence of states lies strictly above the hyperbola of $n_{\mathcal {A}}(\psi)n_{\mathcal {B}}(\psi) = d$ and strictly below the line $n_{\mathcal {A}}+n_{\mathcal {B}}= d+1$. We also show where the holes are when $n_{\mathcal {B}}=2, 3$. Then we present that the KD nonclassicality of a state based on the DFT matrix can be completely characterized by using the support uncertainty relation $n_{\mathcal {A}}(\psi)n_{\mathcal {B}}(\psi)\geq d$. That is, a state $|\psi\rangle$ is KD nonclassical if and only if $n_{\mathcal {A}}(\psi)n_{\mathcal {B}}(\psi)> d$, whenever $d$ is prime or not. This result gives a positive answer to the conjecture in Ref.\cite{Bievre.2021}. There are still some questions left. For example, a state is nonclassical or not when $\mathcal {A},\mathcal {B}$ are MUBs and $n_{\mathcal {A}}(\psi) \leq \frac{d}{2}$ and $n_{\mathcal {B}}(\psi) \leq \frac{d}{2}$. Furthermore, how is the strength of the KD nonclassicality established ?

\begin{acknowledgments}
This work is supported by NSFC (Grant Nos. 11971151, 11901163, 62171264,11871019, 62272208), the Fundamental Research Funds for the Universities of Henan Province (NSFRF220402), Natural Science Foundation of Hebei Province (F2021205001).
\end{acknowledgments}

\appendix
\section{The proof of Lemma \ref{exist_mine}}
\label{Prove exist_mine}

\begin{proof}
Let us first consider the sufficiency. Without loss of generality, suppose that a submatrix $M=(\langle a_{i}|b_{j}\rangle)_{(d-n_{\mathcal {A}}) \times n_{\mathcal {B}}}$ with $n_{\mathcal {A}}\leq i \leq d-1$ and $0\leq j\leq n_{\mathcal {B}}-1$ satisfies the three conditions in Lemma \ref{exist_mine}. Let $S=\mathbb{Z}_{n_{\mathcal {A}}}$ and $T=\mathbb{Z}_{n_{\mathcal {B}}}$.

Note that $\dim\mathcal {H}(S,T)$ is equal to the dimension of nullspace of $M$, that is, $\dim\mathcal {H}(S,T)=n_{\mathcal {B}}-$Rank$(M)$. Then  $\dim\mathcal {H}(S,T)\geq 1$ since Rank$(M)<n_{\mathcal {B}}$. For any $S'\subset S$ for which $|S'|=n_{\mathcal {A}}-1$, assume $S'=S\setminus \{k\}, k \in \mathbb{Z}_{n_{\mathcal {A}}}$. Then $\dim\mathcal {H}(S',T)$=$n_{\mathcal {B}}-$Rank$(M')$, where $M'$ is a $(d-n_{\mathcal {A}}+1) \times n_{\mathcal {B}}$ submatrix of $U$ that is obtained by adding row $k$, i.e., $(\langle a_{k}|b_{0}\rangle,...,\langle a_{k}|b_{n_{\mathcal {B}}-1}\rangle)$, to $M$. By condition (ii), we have Rank$(M')=$Rank$(M)+1$. Thus, $\dim\mathcal {H}(S',T)$=$n_{\mathcal {B}}-$Rank$(M')=n_{\mathcal {B}}-$Rank$(M)-1< \dim\mathcal {H}(S,T)$).

For any $T'\subset T$ for which $|T'|=n_{\mathcal {B}}-1$, assume $T'=T\setminus \{l\}, l\in \mathbb{Z}_{n_{\mathcal {B}}}$. Then $\dim\mathcal {H}(S,T')$ is equal to the dimension of nullspace of $M''$, i.e., $\dim\mathcal {H}(S,T')=n_{\mathcal {B}}-1-$Rank$(M'')$, where $M''$ is a $(d-n_{\mathcal {A}}) \times (n_{\mathcal {B}}-1)$ submatrix of $U$ that is obtained by removing column $l$ of $M$, i.e., $(\langle a_{1}|b_{l}\rangle,...,\langle a_{n_{\mathcal {A}}-1}|b_{l}\rangle)^{T}$. By condition (iii), we have Rank$(M'')=$Rank$(M)$. It follows $\dim\mathcal {H}(S,T')=n_{\mathcal {B}}-1-$Rank$(M'')=n_{\mathcal {B}}-1-$Rank$(M)< \dim\mathcal {H}(S,T)$. Therefore, $(n_{\mathcal {A}}, n_{\mathcal {B}})\in$ UNDC$(\mathcal {A}, \mathcal {B})$ by Lemma \ref{exist_Bievre}.

Now we turn to show the necessity. We proceed by contradiction. If condition (i) does not hold, i.e., Rank$(M) = n_{\mathcal {B}}$, it implies $\dim\mathcal {H}(S,T)=0$. Thus $(n_{\mathcal {A}}, n_{\mathcal {B}})\notin$ UNDC$(\mathcal {A}, \mathcal {B})$ by Lemma \ref{exist_Bievre}. It is a contradiction. If condition (ii) cannot be satisfied, i.e., for any $(d-n_{\mathcal {A}}) \times n_{\mathcal {B}}$ submatrix $M$ of $U$, there exists a row, called row $k$, that is added to $M$ such that Rank$(M')=$Rank$(M)$. Hence, $\dim\mathcal {H}(S',T)= \dim\mathcal {H}(S,T)=n_{\mathcal {B}}-$Rank$(M)$, where $S'=S\setminus \{k\}$. It means $(n_{\mathcal {A}}, n_{\mathcal {B}})\notin$ UNDC$(\mathcal {A}, \mathcal {B})$ by Lemma \ref{exist_Bievre}. Similarly, we can also obtain the desired result when condition (iii) cannot be satisfied.
\end{proof}

\section{The proof of Lemma \ref{properties on DFT}}
\label{Prove properties on DFT}

\begin{proof}
Suppose $\min\{s,t\}=u$. Consider the $u\times u$ submatrix $\widetilde{M}$ in Eq.(\ref{M-(i)}),
\begin{eqnarray}
\label{}
\widetilde{M}=&F\left(
\begin{array}{cccc}
 i_{0}, i_{0}+m, \cdots, i_{0}+(u-1)m;\\
 j_{0}, j_{1}, \cdots, j_{u-1}.
\end{array}
\right)\nonumber\\
=&\left(
\begin{array}{cccc}
 \omega_{d}^{i_{0}j_{0}} & \cdots & \omega_{d}^{i_{0}j_{u-1}}\\
 \cdots &\cdots &\cdots \\\
 \omega_{d}^{i_{0}j_{0}+(u-1)mj_{0}} & \cdots & \omega_{d}^{i_{0}j_{u-1}+(u-1)mj_{u-1}}
\end{array}
\right).\nonumber
\end{eqnarray}
Notice that
\begin{eqnarray}
\label{}
\det(\widetilde{M})&=&\omega_{d}^{i_{0}(j_{0}+\ldots+j_{u-1})}\prod_{0\leq l<k\leq u-1}(\omega_{d}^{mj_{k}}-\omega_{d}^{mj_{l}})\nonumber\\
&=&\omega_{d}^{i_{0}(j_{0}+\ldots+j_{u-1})}\prod_{0\leq l<k\leq u-1}(\omega_{\frac{d}{m}}^{j_{k}}-\omega_{\frac{d}{m}}^{j_{l}})\neq 0\nonumber
\end{eqnarray}
since $\omega_{\frac{d}{m}}^{j_{k}}\neq\omega_{\frac{d}{m}}^{j_{l}} \mod \frac{d}{m}$. Thus Rank($M'$)=$u$. The desired result is obtained.
\end{proof}

\section{The proof of Theorem \ref{nonexistance-nb=3}}
\label{Prove nonexistance-nb=3}

\begin{proof}
Sufficiency can be obtained by Theorem \ref{UNCD-main} and the discussion above Corollary \ref{existance on d+1}. Now we consider the necessary. Since $(d,3)$ belongs to UNDC$(\mathcal {A}, \mathcal {B})$, we have $n=0$. If $n=1$ or $2$, we can take $m=1$ such that the result holds. For $n\geq 3$, a point $(d-n, 3)$ belongs to UNDC$(\mathcal {A}, \mathcal {B})$ of $F$. It means there exists a $n \times 3$ submatrix
\begin{eqnarray}
\label{}
M=&\left(
\begin{array}{cccc}
 \omega_{d}^{i_{0}j_{0}} & \omega_{d}^{i_{0}j_{1}} & \omega_{d}^{i_{0}j_{2}}\\
 \cdots & \cdots & \cdots\\
 \omega_{d}^{i_{n-1}j_{0}} & \omega_{d}^{i_{n-1}j_{1}} & \omega_{d}^{i_{n-1}j_{2}}
\end{array}
\right).\nonumber
\end{eqnarray}
satisfying Lemma \ref{exist_mine}. Then Rank$(M)\leq 2$. Without loss of generality, assume that the first column of $M$ is a linear combination of the other two columns. It follows $\omega_{d}^{i_{k}j_{0}}=c_{1}\omega_{d}^{i_{k}j_{1}}+c_{2}\omega_{d}^{i_{k}j_{2}}$, where $k=0,\ldots,n-1$ and $c_{1},c_{2}$ are the coefficients of the linear combination. Then $c_{1}\omega_{d}^{i_{k}(j_{1}-j_{0})}+c_{2}\omega_{d}^{i_{k}(j_{2}-j_{0})}=1$. By the Euler's formula $e^{x\sqrt{-1}}=\cos x$+$\sqrt{-1}\sin x$, we obtain
\begin{eqnarray}
\label{expand by Eular}
c_{1}\cos \frac{2\pi i_{k}(j_{1}-j_{0})}{d}+c_{2}\cos \frac{2\pi i_{k}(j_{2}-j_{0})}{d} &=& 1,\nonumber\\
c_{1}\sin \frac{2\pi i_{k}(j_{1}-j_{0})}{d}+c_{2}\sin \frac{2\pi i_{k}(j_{2}-j_{0})}{d} &=& 0.
\end{eqnarray}
Calculating the square of two sides of the two equations in Eq.(\ref{expand by Eular}) and then adding the two equations, we obtain
\begin{eqnarray}
\label{c1c2}
2c_{1}c_{2}\cos \frac{2\pi i_{k}(j_{1}-j_{2})}{d} = 1-c_{1}^{2}-c_{2}^{2},
\end{eqnarray}

Case 1. If $c_{1}c_{2}= 0$, then $c_{1} = 0$ or $c_{2} = 0$. Without loss of generality, suppose $c_{1} = 0$. Then $\omega^{i_{k}(j_{0}-j_{2})}=c_{2}$. For $i_{k}\neq i_{l}$, $\omega^{(i_{k}-i_{l})(j_{0}-j_{2})}=1$. It follows $(i_{k}-i_{l})(j_{0}-j_{2}) \equiv 0 \mod d$ for any $k,l\in \mathbb{Z}_{n}$. Let $\gcd(j_{2}-j_{0}, d)=p$. We have $j_{2} \equiv j_{0} \mod p$ and $i_{k} \equiv i_{l} \mod \frac{d}{p}$. It means that $i_{k}$ is in the congruence class of $i_{l}$ modulo $\frac{d}{p}$ for any $k,l\in \mathbb{Z}_{n}$. It implies that $n\leq p$ since the cardinality of the congruence class $i_{l}$ modulo $\frac{d}{p}$ is $p$. Next we will show that it is impossible $n < p$.

Since $c_{1} = 0$, Rank$(M)$ must be 1. Otherwise, if Rank$(M)=2$, we can remove column $j_{1}$ and then Rank$(M'')=1$, where $M''$ is a new submatrix of $F$ that is obtained by removing column $j_{1}$ of $M$. It contradicts condition (iii) of Lemma \ref{exist_mine}.

For any $k,l\in \mathbb{Z}_{n}$, row $i_{k}$ of $M$ is proportional to row $i_{l}$ since Rank$(M)=1$. That is, $\omega_{d}^{(i_{k}-i_{l})j_{0}}=\omega_{d}^{(i_{k}-i_{l})j_{1}}=\omega_{d}^{(i_{k}-i_{l})j_{2}}$. Hence $(i_{k}-i_{l})(j_{0}-j_{1})\equiv 0 \mod d$. It follows $j_{0} \equiv j_{1}\equiv j_{2} \mod p$ since $p$ is prime. If $n < p$, it means that there exists at least one row, say row $i_{s}$, that is not in submatrix $M$ but satisfies $i_{s} \equiv i_{k} \mod \frac{d}{p}$. Thus, row $i_{s}$ can be added to $M$ to obtain submatrix $M'$ such that Rank$(M')=$Rank$(M)$ since $\omega_{d}^{(i_{k}-i_{s})j_{0}}=\omega_{d}^{(i_{k}-i_{s})j_{1}}=\omega_{d}^{(i_{k}-i_{s})j_{2}}$. It contradicts condition (ii) of Lemma \ref{exist_mine}. It implies $n = p$. The required divisor $m$ of $d$ is $p$ satisfying $m=n$ and $3m \leq d$.

Case 2. If $c_{1}c_{2}\neq 0$, from Eq.(\ref{c1c2}) we have
\begin{eqnarray}
\label{}
\cos \frac{2\pi i_{k}(j_{1}-j_{2})}{d} = \frac{1-c_{1}^{2}-c_{2}^{2}}{2c_{1}c_{2}},
\end{eqnarray}
It follows
\begin{eqnarray}
\label{}
\cos \frac{2\pi i_{k}(j_{1}-j_{2})}{d} = \cos \frac{2\pi i_{l}(j_{1}-j_{2})}{d}, k,l \in \mathbb{Z}_{n}.
\end{eqnarray}
Thus, we have
\begin{eqnarray}
\label{equation-n-3}
(i_{k}\pm i_{l})(j_{1}-j_{2}) \equiv 0 \mod d.
\end{eqnarray}
Suppose $\gcd(j_{1}-j_{2}, d)=p$. Then $j_{2} \equiv j_{1} \mod p$ and $i_{k} \equiv \pm i_{l} \mod \frac{d}{p}$. By Eq.(\ref{equation-n-3}) there are two subcases as follows.

Case 2.1. If $i_{k} \equiv i_{l} \mod \frac{d}{p}$ for any $k,l \in \mathbb{Z}_{n}$, by assumption we have
\begin{eqnarray}
\omega_{d}^{i_{k}j_{0}}&=&c_{1}\omega_{d}^{i_{k}j_{1}}+c_{2}\omega_{d}^{i_{k}j_{2}}=\omega_{d}^{i_{k}j_{2}}(c_{1}\omega_{d}^{i_{k}(j_{1}-j_{2})}+c_{2})   \nonumber\\
\omega_{d}^{i_{l}j_{0}}&=&c_{1}\omega_{d}^{i_{l}j_{1}}+c_{2}\omega_{d}^{i_{l}j_{2}}=\omega_{d}^{i_{l}j_{2}}(c_{1}\omega_{d}^{i_{l}(j_{1}-j_{2})}+c_{2}).   \nonumber
\end{eqnarray}
Notice that $\omega_{d}^{i_{k}(j_{1}-j_{2})}=\omega_{d}^{i_{l}(j_{1}-j_{2})}$ since $j_{2} \equiv j_{1} \mod p$ and $i_{k} \equiv i_{l} \mod \frac{d}{p}$. It follows $c_{1}\omega_{d}^{i_{k}(j_{1}-j_{2})}+c_{2}=c_{1}\omega_{d}^{i_{l}(j_{1}-j_{2})}+c_{2}$. Then $\omega_{d}^{(i_{k}-i_{l})j_{0}}=\omega_{d}^{(i_{k}-i_{l})j_{2}}$. It follows $j_{0} \equiv j_{1}\equiv j_{2} \mod p$ since $p$ is prime. A similar discussion in Case 1 can be applied here. We can obtain $n=p=m$.

Case 2.2. Otherwise, there exists a index set $\mathcal {C}\subseteq \mathbb{Z}_{n}$ such that $i_{k} \equiv i_{0} \mod \frac{d}{p}, k\in \mathcal {C}$ and $i_{l} \equiv -i_{0} \mod \frac{d}{p}, l \in \mathbb{Z}_{n}\setminus\mathcal {C}$. Applying a similar discussion in Case 2.1 for $i_{k} \equiv i_{0} \mod \frac{d}{p}, k\in \mathcal {C}$, we have the cardinality of $\mathcal {C}$ is $p$. Similarly, the cardinality of $\mathbb{Z}_{n}\setminus\mathcal {C}$ is also $p$. Thus $n=2p$. The required divisor $m$ of $d$ is $p$ satisfying $2m=n$ and $3m \leq d$.
\end{proof}

\nocite{*}

\bibliography{apssamp}% Produces the bibliography via BibTeX.

\end{document}